\documentclass[10pt]{article}
\usepackage{amssymb,amsmath,amsthm,multirow}
\usepackage{color}
\usepackage[all]{xy}

\newtheorem{example}{Example}
\newtheorem{remark}{Remark}
\newtheorem{theorem}{Theorem}
\newtheorem{proposition}{Proposition}

\newtheorem{corollary}{Corollary}

\newcommand{\NN}{\mathbb{N}}
\usepackage{authblk}

\title{Multivariable codes in principal ideal polynomial quotient rings with applications to additive modular bivariate codes over $\mathbb{F}_4$}
 \author[1]{E. Mart\'{\i}nez-Moro\thanks{Partially supported by the Spanish MINECO under grants MTM2015-65764-C3-1-P and MTM2015-69138-REDT.}}
\author[1]{A. Pi\~nera-Nicol\'as\thanks{Partially supported by MINECO-13-MTM2013-45588-C3-1-P.}}
\author[2]{I.F. R\'ua\thanks{
Partially supported by MINECO-13-MTM2013-45588-C3-1-P and Principado de Asturias Grant GRUPIN14-142.}}

\affil[1]{Mathematics Research Institute (IMUVa), Universidad de Valladolid \\
\tt{Edgar.Martinez@uva.es, alejandro.pinera@uva.es} }

\affil[2]{Departamento de Matem\'aticas, Universidad de Oviedo\\
\tt{rua@uniovi.es}}
\begin{document}

\maketitle

\begin{abstract} 

In this work, we study the structure of multivariable modular codes over finite chain rings when the ambient space is a principal ideal ring. We also provide some applications to additive modular codes over the finite field $\mathbb{F}_4$.



\end{abstract}

\section{Introduction}

Multivariable codes over a finite field $\mathbb{F}_q$ are a natural generalization of several classes of codes, including cyclic, negacyclic, constacyclic, polycyclic and abelian codes. Since these particular families have also been considered over  finite chain rings  (e.g., over  Galois rings), we proposed in \cite{nuestro1,nuestro2} constructions of multivariable codes over them. As with classical cyclic codes over finite fields, the modular case (i.e., codes with repeated roots) is much more difficult to handle than the semisimple case (i.e., codes with non-repeated roots). In this sense, different authors have studied the properties of cyclic, negacyclic, constacylic and polycyclic modular codes over finite chain rings. Among these codes, those contained in an ambient space which is a principal ideal ring admit a relatively simple description, quite close to that of semisimple codes. This feature has been recently used in the description of abelian codes over a finite field \cite{PIGAS}, and in the description of modular additive cyclic codes over $\mathbb{F}_4$ \cite{Huffman2}. As a natural continuation of these works, in this paper we consider the structure of multivariable modular codes over finite chain rings when the ambient space is a principal ideal ring.

\section{Finite chain rings and codes over them}

An associative, commutative, unital, finite ring $R$ is called \emph{chain ring} if it has a unique maximal ideal $M$ and it is principal (i.e, generated by an element $a$). This condition is equivalent \cite[Proposition 2.1]{CN-FCR} to the fact that the set of ideals of $R$ is the chain (hence its name) $\left\langle 0\right\rangle =\left\langle a^t \right\rangle \subsetneq \left\langle a^{t-1} \right\rangle \subsetneq \dots \subsetneq \left\langle a^1 \right\rangle= M \subsetneq \left\langle a^0 \right\rangle=R
$, where $t$ is the nilpotency index of the generator $a$.
The quotient ring $\overline R=R/M$ is a finite field $\mathbb F_q$ where $q= p^l$ is a prime number power. Examples of finite chain rings include Galois rings $GR(p^n,l)$ of characteristic $p^n$ and $p^{nl}$ elements (here $a=p$, and $t=n$) and, in particular, finite fields ($\mathbb{F}_q=GR(p,l)$) \cite{McD}.

\emph{Multivariable codes} over finite chain rings, i.e., ideals of the ring
$\mathcal R=R[X_1,\dots,X_r]/\left\langle t_1(X_1),\dots,t_r(X_r)\right\rangle
$, where $t_i(X_i)\in R[X_i]$ are monic polynomials, were introduced in \cite{nuestro1,nuestro2}. These codes generalize the notion of multivariable codes over a finite field $\mathbb{F}_q$, as presented in \cite{Poli}, and include well-known families of codes over a finite chain ring alphabet. For instance \emph{cyclic} ($r=1,t_1(X_1)=X_1^{e_1}-1$), \emph{negacyclic} ($r=1,t_1(X_1)=X_1^{e_1}+1$), \emph{constacyclic} ($r=1,t_1(X_1)=X_1^{e_1}+\lambda$), \emph{polycyclic} ($r=1$) and \emph{abelian codes} ($t_i(X_i)=X_i^{e_i}-1,\forall i=1,\dots,r$) \cite{CN-FCR,SergioSteve}. Properties of multivariable codes over a finite chain ring depend on the structure of the ambient ring $\mathcal R$. So, in \cite{nuestro1} a complete account of codes was given when the polynomials $\overline{t_i}(X_i)\in \mathbb{F}_q[X_i]$ have no repeated roots (the so-called \emph{semisimple} or \emph{serial} case). On the other hand, as a first approach to the \emph{repeated-root} (or \emph{modular}) case,  Canonical Generating Systems \cite{CGS} were considered  in \cite{nuestro2}. Unfortunately, the description is not as satisfactory as in the semisimple case. This situation agrees with that of cyclic, negacyclic, constacylic and polycyclic repeated-root codes. Different authors have dedicated their efforts to provide a better understanding of these codes over finite chain rings
(see, for instance \cite{modularcyclic,Blackford,Salagean,SergioSteve}). 
 
 One important feature of semisimple codes is that all of them can be generated by a single codeword, i.e., they can be regarded as principal ideals in $\mathcal R$. This property is not generally true in the modular case, and it partly explains the reason why these codes are more difficult to describe. However, that  of all the ideals in $\mathcal R$ are principal is not equivalent to the semisimple condition. Instead, it is equivalent to the fact that its nilradical is principal \cite[Lemma 3]{Cazaran}. As it was shown in \cite[Theorem 2]{Cazaran}, we have the following characterization (see also \cite[Theorem 5.2]{SergioSteve}, \cite[Theorem 3.2]{Salagean}, \cite[Theorem 1]{nuestro2}).
 
\begin{theorem}\label{PIR}
    The ring $\mathcal R=R[X_1,\dots,X_r]/\left\langle t_1(X_1),\dots,t_r(X_r)\right
    \rangle$ is a \emph{principal ideal ring (PIR)} if and only if one of the following conditions is satisfied:
    \begin{enumerate}
    \item If $R$ is a Galois ring $GR(p^n,l)$, then the number  of
    polynomials for which $\overline t_i(X_i)$ is not
    square-free is at most one. Moreover, if $R$ is not a finite field (i.e., $n>1$), and
    $\overline t_i(X_i)$ is not square-free with
    $$t_i(X_i)=g(X_i)h(X_i)+au(X_i)$$ where $\overline g(X_i)$ is
    the square-free part of $\overline t_i(X_i)$, then
$\overline u(X_i)$ and $\overline h(X_i)$ are coprime polynomials.
	\item If $R$ is not a Galois ring, then $r=1$, and $\overline t_1(X_1)$ is square-free.

     \end{enumerate}
\end{theorem}

\begin{example}\label{ejemplo1} Let us consider the ring $R=\mathbb{Z}/4\mathbb{Z}$, which is the Galois ring  $GR(4,1)$, and the polynomials $t_1(X_1)=X_1^2+1$ and $t_2(X_2)=X_2^7-1$. Following Theorem \ref{PIR}, $t_1(X_1)$ can be written as $t_1(X)=(X_1-1)^2+2X_1$. Since $\overline{h}(X_1)=X_1+1$ and $\overline{u}(X_1)=X_1$ are coprime polynomials, then $\mathcal R=R[X_1,X_2]/\langle X_1^2+1,X_2^7-1\rangle$ is a principal ideal ring. Notice that the ring $R[X_1]/\left<X_1^2+1\right>$  is also a principal ideal ring {and its ideals are negacyclic codes.}

%\textcolor{blue}{Espacio ambiente es generalizaci�n del ejemplo 3.2 del articulo de Sergio. Tal vez podamos tirar de las secciones 6  y 7 de Sergio. Self-dual (para los cu�nticos) y relacionar con los PIGAS.}
\end{example}
 
The principal ideal property has been recently used in the description of modular abelian codes over a finite field \cite{PIGAS}, and in the description of modular \emph{additive} cyclic codes over $\mathbb{F}_4$ (i.e., additive subgroups of the ambient ring $\mathbb{F}_4[X_1]/\left\langle X_1^{e_1}-1\right\rangle$, $e_1$ even) \cite{Huffman2}. As a natural continuation of these works, in this paper we consider the structure of multivariable modular codes over finite chain rings when the ambient space is a principal ideal ring. Since this ring is a polynomial quotient ring we will call it a \emph{principal ideal polynomial quotient ring (PIPQR)}.  Our aim is to achieve a complete description of them and their properties.

\section{Multivariable modular codes in PIPQRs}

From now on we will restrict our attention to multivariable codes over a finite chain ring in an ambient space $\mathcal R$ which is a PIR, i.e., to multivariable codes in PIPQRs. We will impose the modular (or repeated-root) condition, as the semisimple case was fully treated in \cite{nuestro1}. Hence, from Theorem \ref{PIR}, $R$ must be a Galois ring $GR(p^n,l)$, and there exists exactly one index $i=1,\dots,r$ such that $\overline t_i(X_i)$ has repeated-roots. We shall assume w.l.o.g. that $i=1$. Let $\overline t_1(X_1)=\prod_{j=1}^s\overline{g_j}(X_1)^{k_j}$ be the unique decomposition in powers of coprime irreducible polynomials $\overline{g_j}(X_1)\in\mathbb{F}_q[X_1]$ of degree $r_j$. Then, because of Hensel's lemma \cite[Theorem XIII.4]{McD}, there exist monic $G_j(X_1)\in R[X_1]$ pair-wise coprime polynomials such that $t_1(X_1)=\prod_{j=1}^s G_j(X_1)$, and $G_j(X_1)=g_j(X_1)^{k_j}+pu_j(X_1)$ (i.e. $t_1(X_1)$ is decomposed as a product of primary coprime polynomials).
As in \cite[Section 5]{SergioSteve} we may assume w.l.o.g. that $g_j(X_1)$ is monic and $r_jk_j>\deg{u_j(X_1)}$.

Hence, $g(X_1)=\prod_{j=1}^s g_j(X_1)$ is such that $\overline{g}(X_1)$ is the square-free part of $t_1(X_1)=g(X_1)h(X_1)+pu(X_1)$, where $h(X_1)=\prod_{j=1}^s g_j(X_1)^{k_j-1}$, and $u(X_1)=\sum_{j=1}^s u_j(X_1)\prod_{l\not= j}g_l(X_1)^{k_l}+p\Delta(X_1)$, for some $\Delta(X_1)\in R[X_1]$. If $R$ is not a finite field, then the principal condition is equivalent to $\overline{u}(X_1)=\sum_{j=1}^s \overline{u_j}(X_1)\prod_{l\not= j}\overline{g_l}(X_1)^{k_l}$ nonzero and coprime with 
$\overline{h}(X_1)=\prod_{j=1}^s \overline{g_j}(X_1)^{k_j-1}$.  This means that $\overline{g_j}\ \not|\ \overline{u_j}$, whenever $k_j\ge 2$.

A first question that can be asked is whether it is possible to obtain analogues of \emph{abelian codes in principal ideal group algebras} \cite{PIGAS} in this setting. The answer is not, as the following corollary of Theorem \ref{PIR} shows.

\begin{corollary}\label{noabeliano}
    If $\mathcal R=R[X_1,\dots,X_r]/\left\langle X_1^{e_1}-1,\dots,X_r^{e_r}-1\right
    \rangle$ is a {principal ideal polynomial quotient ring (PIPQR)}, then either $R$ is a finite field $\mathbb{F}_q$ (i.e., $\mathcal R$ is a principal ideal group algebra, PIGA) or we are in the semisimple case.
\end{corollary}

\begin{proof}
If we are not in the semisimple case, then $R$ is a Galois ring, $\overline t_1(X_1)=X_1^{e_1}-1$ has repeated-roots, and so $e_1=p^{l_1}m_1$, with $l_1\ge 1$ and $p_1\ \not|\ m_1$. Besides, if $R$ is not a finite field, then $t_1(X_1)=g(X_1)h(X_1)+pu(X_1)$, with $g(X_1)=X_1^{m_1}-1,h(X_1)=1+X_1^{m_1}+\dots+X_1^{(p^{l_1}-1)m_1},u(X_1)=0$. But this is a contradiction with Theorem \ref{PIR}, because $\overline{g}(X_1)$ is the square-free part of $\overline t_1(X_1)$, and $\overline h(X_1),\overline u(X_1)$ are not coprime polynomials.
\end{proof}

As the abelian case is fully treated in \cite{PIGAS} we will impose the condition that $\mathcal R$ is not a group ring henceforth. Even though we cannot have abelian codes in our setting it is interesting to mention how abelian codes in PIGAs are viewed in \cite{PIGAS}. In such a paper the ring $\mathcal R$ is a group algebra $\mathbb{F}_q[G]$ over an abelian finite group $G$ which is a direct product of a cyclic $p-$Sylow $B$ and a complementary subgroup $A$. A one-to-one correspondence between the ring $\mathbb{F}_q[G]$ and the group ring $\mathcal S[B]$ (where $S=\mathbb{F}_q[A]$ is a semisimple group ring) is used to describe all the codes (i.e., ideals) in the former ring (see sections II and III in \cite{PIGAS}). We want to use the same type of approach in our case: adjoin the semisimple part to the base ring and use its decomposition as sum of finite chain rings \cite{nuestro1} to describe the original PIPQR. Let us begin this technique with the univariable case, i.e., with the description of polycyclic codes. 

\begin{proposition}\label{univariado}
Let $\mathcal R=R[X_1]/\left\langle t_1(X_1)\right\rangle$ be a PIPQR such that $\overline{t_1}(X_1)\in \mathbb{F}_q[X_1]$ has repeated-roots (in particular, $R$ is a Galois ring $GR(p^n,l)$). Then, $\mathcal R$ is a {direct sum} of finite chain rings $\mathcal R_j$. {For each of these rings, the maximal ideal  has nilpotency index $nk_j$ and the residual field $\overline{\mathcal R_j}$ is the finite field $\mathbb{F}_{q^{r_j}}$ ($q=p^l$).}
\end{proposition}

\begin{proof}
%Notice first that, because of Corollary \ref{noabeliano}, if $R\not=\mathbb{F}_q$, then $t_1(X_1)\not=X_1^e-1$ for all $e\in \NN$. 
We will follow \cite{McD} through \cite{SergioSteve} (ambient structure of polycyclic codes). With the previous notation $\mathcal R= \bigoplus_{j=1}^s\mathcal R_j$, where 
$\mathcal R_j=\left\langle\widetilde{G_j}(X_1)\right \rangle\cong R[X_1]/\left\langle{G_j}(X_1)\right \rangle$ and $\widetilde{G_j}(X_1)=\prod_{l\not =j}G_l(X_1)$
 \cite[Theorem 5.1]{SergioSteve}. Now, with the same argument of \cite[Theorem 3.2]{Salagean} we can say that $\mathcal R_j$ is local PIR with maximal ideal $\left\langle p\widetilde{G_j}(X_1)\ ,\ {g_j}(X_1)\widetilde{G_j}(X_1)\right\rangle=\left\langle {g_j}(X_1)\widetilde{G_j}(X_1)\right\rangle$. 
 %By hypothesis $\mathcal R$ is a PIR, and so it is $\mathcal R_j$ too. Notice that if $v$ is a generator of the nilradical of $\mathcal R$, then $v_j=v\widetilde{
%G_j}(X_1)$ is a generator of the nilradical of $\mathcal R_j$. 
Therefore $\mathcal R_j$ is a finite chain ring and
 $\overline{\mathcal R_j}\cong \mathbb{F}_q[X_1]/\left\langle\overline{g_j}(X_1)\right\rangle$ is a field extension of $\mathbb{F}_q$ of degree $r_j$. Because of the proof of
 \cite[Lemma XVII.4]{McD}, $|\mathcal R_j|=\left(q^{r_j}\right)^{w_j}$, where $w_j$ is the nilpotency index of ${g_j}(X_1)\widetilde{G_j}(X_1)$. But since $G_j(X_1)$ is a monic polynomial of degree $k_jr_j$, we have that $|\mathcal R_j|=\left(q^{n}\right)^{k_jr_j}$, and so $w_j=nk_j$.
 
% Because the nilpotency index of $p\widetilde{G_j}(X_1)$ is $n$, and because 
% $\left\langle g_j(X_1)^{k_j}\widetilde{G_j}(X_1)\right\rangle= \left\langle p\widetilde{G_j}(X_1)\right\rangle$ \cite[Theorem 3.2]{Salagean}, we get $(n-1)k_j<w_j\le nk_j$, and if $k_j=1$, then $w_j=n$ .
\end{proof}

%\begin{remark}Because of \cite[Lemma XVII.4]{McD} we know that the nilpotency index $w_j$ satisfies the bounds $(n-1)s_j+1\le w_j\le ns_j$, where $p^n$ is the characteristic of $R$, and $s_j$ is the smallest natural number such that $p\in \left\langle v_j^{s_j}\right\rangle$. Hence, $s_j\le\frac{(k_j+1)t-2}{n-1}$.
%\end{remark}

%\textcolor{red}{No tengo muy claro que la nota sea muy informativa ni tenga interes pero es por conectarlo con lo que comentamos hoy por la mangana.}

\begin{remark}\label{ideales} In view of \cite[Theorem XVII.5]{McD} we have the following description of the ambient space ring $\mathcal R$ as a direct sum of finite chain rings (cf. \cite[Equation (II.5)]{PIGAS}):
$$\mathcal R\cong \bigoplus_{j=1}^s R_j[X_1]/\left\langle \gamma_j(X_1)\ ,\ p^{n-1}X_1^{k_j}\right\rangle$$
where $R_j=GR(p^n,lr_j)$,and $\gamma_j(X_1)\in R_j[X_1]$ is an Eisenstein polynomial of degree $k_j$ of the form $X_1^{k_j}+pf_j(X_1)$. Moreover, for each factor the set of nonzero ideals is 
$\left\{\left\langle p^iX_1^l \right\rangle\ |\ 0\le i\le n-1\ ,\ 0\le l\le k_j-1 \right\}$
\end{remark}

Now, let us describe all possible univariable codes (cf. \cite[Corollaries 3.11, 3.12]{nuestro1}).

%So, if $\mathcal K$ is a multivariable code in ... it is uniquely be generated by... (hacker referencia a nuestros resultados del case modular).

\begin{corollary}\label{cor:des-ideales}  {If
$\mathcal R=R[X_1]/\left\langle t_1(X_1)\right\rangle$ is a PIPQR such that $\overline{t_1}(X_1)\in \mathbb{F}_q[X_1]$ has repeated-roots (in particular, $R$ is a Galois ring $GR(p^n,l)$),
then any code $\mathcal K\vartriangleleft
\mathcal R$ is a  sum of ideals of the form 
$$
\left\langle p^{i_j}g_j(X_1)^{c_j}\widetilde{G_j}(X_1)\right\rangle,$$
where $(i_j,c_j)=(n,0)$ or $ 0\le i_j\le n-1\ ,\
0\le c_j\le k_j-1\ ,\
1\le j\le s.
$
Hence, there exists a family of
 polynomials  $H_1,\dots,H_{n}\in R[X_1]$
 such that
 \begin{equation}
\mathcal K=\left\langle  H_1,p H_2,\dots,p^{n-1}
 H_{n}\right\rangle=\left\langle
 \sum_{i=0}^{n-1}p^iH_{i+1}\right\rangle
 \end{equation}

Moreover, $$|\mathcal K|=|\bar R|^{\sum_{j=1}^{s}r_j(nk_j-c_j-i_jk_j)} $$
and there exist $\prod_{j=1}^s(nk_j+1)$ repeated-root codes in $\mathcal R$.}

\end{corollary}

\begin{example}\label{ejemplo2} (Example \ref{ejemplo1} cont'd). In the special case of $R[X_1]/\langle X_1^2+1\rangle$, we have $g_1(X_1)=X_1-1$, $G_1(X_1)=t_1(X_1)$ and $\widetilde{G_1}=1$. Moreover, since $(X_1-1)^2\equiv 2X_1\emph{ mod }t_1(X_1)$ and $X_1$ is a unit in $\mathcal R$, the ideal $\langle (X_1-1)^2\rangle$ is equal to $\langle 2\rangle$. So, $R[X_1]/\langle X_1^2+1\rangle$ is a finite chain ring with ideals (i.e.{negacyclic} codes)

$$R[X_1]/\langle X_1^2+1\rangle \supsetneq  <X_1-1>\supsetneq  <2> \supsetneq  <2(X_1-1)> \supsetneq  0.$$
\end{example}

In the following result we adjoin the semisimple part to the ring $R$ in order to describe the original PIPQR with the help of Proposition \ref{ideales}.

\begin{theorem}\label{teoremon}
Let  $\mathcal R=R[X_1,\dots,X_r]/\left\langle t_1(X_1),\dots,t_r(X_r)\right
    \rangle$ be a PIPQR such that $\overline{t_1}(X_1)\in \mathbb{F}_q[X_1]$ has repeated-roots (in particular, $R$ is a Galois ring $GR(p^n,l)$). Then $\mathcal{R}$ is a direct sum of finite chain rings $R_{C,j}[X_1]/\left\langle \gamma_j(X_1),p^{n-1}X_1^{k_j}\right\rangle$, where $R_{C,j}$ is a Galois extension of the coefficient ring $R$ and $\gamma_j(X_1)$ is an Eisenstein polynomial of degree $k_j$.
\end{theorem}

\begin{proof}
The ring $\mathcal R$ is isomorphic to the tensor product 
$$R[X_1]/\langle t_1(X_1)\rangle \otimes \left( R[X_2,\ldots,X_r]/\langle t_2(X_2),\ldots,t_r(X_r)\rangle\right).$$
Since $t_2(X_2),\cdots,t_r(X_r)$ have simple roots only, from \cite[Theorem 3.9]{nuestro1}, there exists an isomorphism
$$\varphi: \bigoplus_{C\in\mathcal C} Q_C\longrightarrow R[X_2,\ldots,X_r]/\langle t_2(X_2),\ldots,t_r(X_r)\rangle$$
where 
\begin{equation}\label{eq:C}
{\mathcal C}=\{\{(\mu_2^{q^s},\dots,\mu_r^{q^s})\ |\ s\in\NN\}\ |\ \overline{t_i}(\mu_i)=0,\ i=2,\dots,r\}
\end{equation}
is the set of all {cyclotomic classes} of the roots of $\overline{t_2}(X_2),\cdots,\overline{t_r}(X_r)$. Each $Q_C=GR(p^n,l\cdot|C|)$ is a Galois extension of $R$ contained in $R[X_2,\ldots,X_r]/\langle t_2(X_2),\ldots,t_r(X_r)\rangle$. Then, since $t_1(X_1)\in R[X_1]$,   $\varphi$ induces an isomorphism 

$$\widehat{\varphi}:\bigoplus_{C\in\mathcal C} Q_C[X_1]/\langle t_1(X_1)\rangle\longrightarrow\mathcal R.$$
As a consequence of Proposition \ref{univariado}, $\mathcal R$ can be written as a direct sum of finite chain rings  $R_{C,j}[X_1]/\left\langle \gamma_j(X_1)\ ,\ p^{n-1}X_1^{k_j}\right\rangle$, where $R_{C,j}$ is a Galois extension of $Q_C$ and so, of $R$.
\end{proof}

Now, we can generalize Corollary \ref{cor:des-ideales} to the multivariable case.

\begin{corollary}\label{cor:des-mult}  If
$\mathcal R=R[X_1,\ldots,X_r]/\left\langle t_1(X_1),\ldots,t_r(X_r)\right\rangle$ is a PIPQR such that $\overline{t_1}(X_1)\in \mathbb{F}_q[X_1]$ has repeated-roots (in particular, $R$ is a Galois ring $GR(p^n,l)$). Let us suppose that for each $C\in\mathcal C$ (see equation (\ref{eq:C})),  $$t_1(X_1)=\prod_{j=1}^s\prod_{m=1}^{s_{j,C}}G_{j,m}^C(X_1)$$
is the decomposition of $t_1(X_1)$ as the product of primary coprime polynomials in $Q_C[X_1]$ (i.e. $\overline{G_{j,m}^C}(X_1)=\overline{g_{j,m}^C}(X_1)^{k_j}$, where $\overline{g_{j,m}^C}(X_1)\in\mathbb F_{q^{|C|}}[X_1]$ is irreducible).
Then any code $\mathcal K\vartriangleleft
\mathcal R$ is {a sum of ideals} of the form
$$
\widehat{\varphi}\left(
%\prod_{C\in\mathcal C}
\left\langle p^{i_{j,m}}{g_{j,m}^C(X_1)^{c_{j,m}}\widetilde{G_{j,m}^C}}(X_1)\right\rangle\right),$$ 
where $(i_{j,m},c_{j,m})=(n,0)$ or
$0\le i_{j,m}\le n-1\ ,\
0\le c_{j,m}\le k_j-1\ ,  \ 1\leq m\leq s_{j,C}\ ,\
1\le j\le s\ ,\ C\in \mathcal C.
$
 Hence, there exists a family of
 polynomials  $H_1,\dots,H_{n}\in R[X_1,\ldots,X_r]$
 such that
 \begin{equation}
\mathcal K=\left\langle  H_1,p H_2,\dots,p^{n-1}
 H_{n}\right\rangle=\left\langle
 \sum_{i=0}^{n-1}p^iH_{i+1}\right\rangle
 \end{equation}

Moreover, $$|\mathcal K|=|\bar R|^{\sum_{C\in \mathcal C}\sum_{j=1}^{s}\sum_{m=1}^{s_{j,C}}\deg{g_{j,m}^C}(nk_j-c_{j,m}-i_{j,m}k_j)} $$
%$$|\mathcal K|=|\bar R|^{\sum_{C\in \mathcal C}\sum_{j=1}^{s}\left(r_jnk_j-\sum_{m=1}^{s_{j,C}}\deg{g_{j,m}^C}(c_{j,m}+i_{j,m})\right)} $$
and there exist $\prod_{C\in \mathcal C}\prod_{j=1}^s(nk_j+1)^{s_{j,C}}$ repeated-root codes in $\mathcal R$.
\end{corollary}

The results on the Hamming distance of linear codes (and in particular of serial multivariable codes) over finite chain rings contained in \cite[Section 2.1]{Nechaev99}, \cite[Section 4]{SN} and \cite[Theorem 2]{nuestro2} can be applied in this context to multivariable codes in PIPQRs. However, because of the special simple description of these codes, we can use the same ideas of \cite[Section 3.3]{nuestro1} to compute their Hamming distance.
 
 \begin{proposition}\label{codigocociente}
    In the conditions of Corollary \ref{cor:des-mult}, $d(\mathcal K)=d(\mathcal L)$, where $\mathcal L$ is the code $\left\langle \overline{H_1},\dots,\overline{H_t}\right\rangle$ in $\mathbb F_q[X_1,\dots,X_r]/\left\langle \overline t_1(X_1),\dots,\overline t_r(X_r)\right\rangle$.
\end{proposition}

Hence, the results on the Hamming distance of the codes $\mathcal L$ (i.e, on multivariable codes with repeated-roots over a finite field) which can be found in \cite{Poli} can be lifted to our codes.

\begin{example}\label{Ex3} (Example \ref{ejemplo1} cont'd). The factorization of $X_2^7-1=(X_2-1)(X_2^3+2X_2^2+X_2+3)(X_2^3+3X_2^2+2X_2+3)$ into basic irreducible polynomials over $\mathbb Z/4\mathbb{Z}$ provides the following decomposition of $\mathcal R$ (see Theorem \ref{teoremon})
$$\mathcal R\cong R[X_1]/\langle X_1^2+1\rangle\oplus GR(4,3)[X_1]/\langle X_1^2+1\rangle\oplus GR(4,3)[X_1]/\langle X_1^2+1\rangle.$$
Each summand is a finite chain ring (cf. Example \ref{ejemplo2}), and so any {negacyclic} code can be written as the direct sum of three ideals.

Let us consider the code $\mathcal K=\langle (X_1-1)(X_2^4+X_2^3-3X_2^2+2X_2+3)\rangle$. The polynomial $X_2^4+X_2^3-3X_2^2+2X_2+3$ is, up to units, an orthogonal idempotent of $\mathcal R$. Namely, it generates the third summand of the previous decomposition. Since $(X_1-1)^2\equiv 2X_1\mbox{ mod }t_1(X_1)$ and $X_1$ is a unit of $\mathcal R$, we deduce that $\mathcal K=\langle H_1, 2H_2\rangle$ with $H_1=(X_1-1)(X_2^4+X_2^3-3X_2^2+2X_2+3)$ and $H_2=X_2^4+X_2^3-3X_2^2+2X_2+3$ (see Corollary \ref{cor:des-mult}). Thus, the code $\langle  \overline{H_1},\overline{H_2}\rangle$ contains a codeword with Hamming weight 4 and, by Proposition \ref{codigocociente}, the Hamming distance of $\mathcal K$ is at most 4. Direct computations with Sage \cite{sage} show that this is the actual minimum distance of the code. On the other hand, observe that the Hamming distance of the code $\overline {\mathcal K}=\langle (X_1+1)(X_2^4+X_2^3+X_2^2+1)\rangle\lhd \mathbb{Z}/2\mathbb{Z}[X_1,X_2]/\left\langle (X_1+1)^2,X_2^7+1\right\rangle$ is 8 (to see this, apply the isometry of \cite[Proposition 45]{Poli} and check \cite[Table 1]{Castagnoli}).

\end{example}

\begin{example}
As a variation of the previous example, let us take the same ambient space $\mathcal R$ and construct the code $\mathcal L$ generated by the codeword $(X_1-1)(X_2^3+2X_2^2+X_2-1)$. This code can be seen as a product code of the negacyclic code generated by $X_1-1$ and the (punctured) $\mathbb{Z}_4-$base linear code of the Kerdock code of length 16.
\end{example}

\begin{example} {(Generalized Kerdock code) Let $R=\mathbb{Z}/4\mathbb{Z}$ and $S=GR(2^2,m)$, with $m$ prime, a Galois extension of $R$. Let $U=1+2R=\langle \eta_1\rangle$ be the group of units of $R$. Let $\lambda\in S$ be a generator of $\Gamma^*(S)$, with $\Gamma(S)$ the Teichm\"uller coordinate set of $S$, and let $\mbox{Tr}:S\rightarrow R$ be the trace function of $S$ onto $R$. According to \cite{Nechaev_95},  
$$\mathcal{K}=\left\{\sum_{i_1=0}^{2^m-1}\sum_{i_2=0}^1((\mbox{Tr}\,(\xi\lambda^{i_1})+\beta)\eta^{i_2})X_1^{i_1}X_2^{i_2}\, |\, \xi\in S,\beta\in R\right\}$$
is an ideal of the ambient space $R[X_1,X_2]/\langle X_1^{2^m-1}-1,X_2^2-1\rangle$ known as generalized Kerdock code. We can regard the ambient space as a direct sum of rings of the form $T[X_2]/\langle X_2^2-1\rangle$, with $T$ a suitable Galois extension of $R$. Such rings are not principal ideal rings, since their nilradical, $\langle 2,X_2+1\rangle$, is not principal (see \cite[Proposition 4.4]{Sergio-Steve}). Thus, the generalized Kerdock code is not a PIPQR. Also, notice that $X_2^2-1=(X_2-1)^2-2(X_2-1)$ and so, statement 1 of Theorem \ref{PIR} is not satisfied.}

\end{example}

\section{Additive modular codes over $\mathbb{F}_4$ from PIPQRs}

Additive modular codes over $\mathbb F_4$ {can be} used {to construct} quantum error correcting codes, as shown in \cite{nuestro5}. However, except in {very special cases (for instance, the} univariable cyclic modular codes annalyzed in \cite{Huffman2}, the description of such codes seems quite difficult. {As an application of the study of the multivariable codes of the previous section, we obtain a complete description of additive modular codes in PIPQRs}.

In the framework presented in \cite{nuestro5}, additive modular codes can be seen as $\mathcal A_2$-additive submodules of the algebra $\mathcal A_4$, where  $\mathcal A_q=\mathbb F_q[X_1,\ldots,X_r]/\langle t_1(X_1),\ldots,t_r(X_r)\rangle$ and $t_1(X_1),\ldots,t_r(X_r)$ have coefficients in $\mathbb{F}_2$. Since the finite field $\mathbb F_q$, with $q=2,4$, is the Galois ring {$GR(2,r)$} with $r=1,2$ respectively, the algebra $\mathcal A_q$ is a PIPQR if and only if the polynomials $t_2(X_2),\ldots,t_r(X_r)$ are square-free. In such {a} case, from Corollary \ref{cor:des-mult}, the algebra $\mathcal A_q$ can be decomposed into a direct sum of ideals (see also \cite[Theorem 1]{nuestro5}). Let us recall some definitions and results from \cite{nuestro3}  in order to describe such a decomposition. 

{The set $\mathcal C_2$ of $2-$classes of the roots of ${t_2}(X_2),\cdots,{t_r}(X_r)$ (take $q=2$ in equation (\ref{eq:C})) is a disjoint union of two subsets according to their relation to the set $\mathcal C_4$ of $4-$classes (take $q=4$ in the same equation). The first subset, $\mathcal C_2^o$, contains the classes $C_2(\mu)$ such that $C_2(\mu)\in \mathcal C_4$, i.e., classes with odd cardinality. The second subset, $\mathcal C_2^e$, contains the classes that split in $\mathcal C_4$, i.e, those $C_2(\mu)$ with even cardinality such that $C_2(\mu)=C_4(\mu)\cup C_4(\mu^2)$}.

\begin{theorem}\label{Aq}
Let $\mathcal A_q=\mathbb F_q[X_1,\ldots,X_r]/\langle t_1(X_1),\ldots,t_r(X_r)\rangle$, $q=2,4$, be a PIPQR such that $t_i(X_i)\in \mathbb{F}_2[X_i]$ for all $i=1,\dots,r$, and such that {only} ${t_1(X_1)}$ has repeated-roots. Let us suppose that for each $C\in\mathcal C_2$ the polynomial $t_1(X_1)$ factorizes as the product of primary coprime polynomials  in {$\mathbb{F}_{2^{|C|}}[X_1]$} as (cf. Corollary \ref{cor:des-mult})  $$t_1(X_1)=\prod_{j=1}^s\left(\prod_{m=1}^{s_{j,C}}g_{j,m}^C(X_1)\right)^{k_j}$$
%is the decomposition of $t_1(X_1)$  (i.e. $\overline{G_{j,m}^C}(X_1)=\overline{g_{j,m}^C}(X_1)^{k_j}$, 
Then: 
\begin{enumerate}
\item $\mathcal A_2$ is a direct sum of finite chain rings $\mathcal K^C_{j,m}\cong \mathbb{F}_{2^{|C|\deg g_{j,m}^C}}[Z]/\langle Z^{k_{j}}\rangle$. 
%where 
%$$\langle h_{j,m}^C(X_1)\rangle\cong \mathbb F_q(\mu_{j,m},\mu_C)[X_1]/\langle g_{j,m}^C(X_1)^{k_{j}}\rangle\cong\mathbb F_q(\mu_{j,m},\mu_C)[Z_1]/\langle Z_1^{k_{j}}\rangle,$$
%where $C\in \mathcal C_2$.% and $1\leq m\leq s_{j,C}$, $1\leq j\leq s$.
\item $\mathcal A_4$ is a direct sum of ideals $\mathcal I^C_{j,m}$ which are free $\mathcal K^C_{j,m}-$modules of rank 2. %($C\in \mathcal C_2$).
%where 
%$$\langle h_{j,m}^C(X_1)\rangle\cong \mathbb F_q(\mu_{j,m},\mu_C)[X_1]/\langle g_{j,m}^C(X_1)^{k_{j}}\rangle\cong\mathbb F_q(\mu_{j,m},\mu_C)[Z_1]/\langle Z_1^{k_{j}}\rangle,$$
% and $1\leq m\leq s_{j,C}$, $1\leq j\leq s$.
\item Any additive modular code $\mathcal D$ is direct sum of subcodes $\mathcal D_{j,m}^C$ which are $\mathcal K_{j,m}^C-$submodules of $\mathcal I_{j,m}^C$.
\end{enumerate}
\end{theorem}

\begin{proof}
\begin{enumerate}
\item This is a direct consequence of Proposition \ref{univariado} and Theorem \ref{teoremon}.
\item The proof depends on the cardinality of each class $C\in \mathcal C_2$, and the degree of the polynomial $g^C_{j,m}(X_1)$.
\begin{enumerate}
\item If $C\in \mathcal C^o_2$ and $g^C_{j,m}(X_1)$ has odd degree, then $C\in \mathcal C_4$ and $g^C_{j,m}(X_1)$ is also irreducible in $\mathbb{F}_{4^{|C|}}[X_1]$. Therefore, because of Proposition \ref{univariado}, there exists an ideal $\mathcal I^C_{j,m}\cong \mathbb{F}_{4^{|C|\deg g^C_{j,m}}}[Z]/\langle Z^{k_{j}}\rangle$ in $\mathcal A_4$, which is clearly a free $\mathcal K^C_{j,m}-$module of rank 2.
\item If $C\in \mathcal C^o_2$ and $g^C_{j,m}(X_1)$ has even degree, then $C\in \mathcal C_4$ and $g^C_{j,m}(X_1)$ splits as the product of two irreducible polynomials $g^C_{j,m,1}(X_1),g^C_{j,m,2}(X_1)$ of the same degree $\frac{1}{2}\deg g^C_{j,m}$ in $\mathbb{F}_{4^{|C|}}[X_1]$. Therefore, because of Proposition \ref{univariado}, there exists an ideal in $\mathcal A_4$
$$\mathcal I^C_{j,m}= \mathcal I^C_{j,m,1}\oplus \mathcal I^C_{j,m,2}\cong \left(\mathbb{F}_{4^{|C|\frac{1}{2}\deg g^C_{j,m}}}[Z]/\langle Z^{k_{j}}\rangle\right)^2$$
%\oplus \mathbb{F}_{4^{\frac{|C|}{2}\deg g^C_{j,m}}}[Z]/\langle Z^{k_{j}}\rangle$$ 
which can be seen as a free $\mathcal K^C_{j,m}-$module of rank 2 (cf. \cite[Proposition 3]{nuestro5}).
\item If $C\in \mathcal C^e_2$, then $C=D\cup E$, with $D,E\in \mathcal C_4$, and $|D|=|E|=\frac{|C|}{2}$. Hence, $g^C_{j,m}(X_1)$ is also irreducible in $\mathbb{F}_{4^{|D|}}[X_1]$, and so  there exists an ideal in $\mathcal A_4$
$$\mathcal I^C_{j,m}= \mathcal I^D_{j,m}\oplus \mathcal I^E_{j,m}\cong \left(\mathbb{F}_{4^{\frac{|C|}{2}\deg g^C_{j,m}}}[Z]/\langle Z^{k_{j}}\rangle\right)^2$$
which again is a free $\mathcal K^C_{j,m}-$module of rank 2.
\end{enumerate}
\item The proof is similar to \cite[Theorem 2]{nuestro3}.
\end{enumerate}
\end{proof}

{Let us illustrate this theorem with a concrete example. It provides, via \cite[Theorem 2]{calderbank}, a way to construct a quantum-error-correcting code with parameters $[[8,4,2]]$. This code has an optimal distance for its length and dimension according to \cite{Grassl}, and it can be fully described as an additive modular code in a PIPQR as we shall see now.}

\begin{example}
{Consider the binary polynomials $t_1(X_1)=(X_1+1)^2,t_2(X_2)=X_2^2+X_2+1$. Theorem \ref{Aq} give us an isomorphism $\mathcal A_2\cong \mathbb{F}_4[Z]/\langle Z^4\rangle $ (here $\mathbb{F}_4=\mathbb{F}_2[X_2],$ and $Z=X_1+1$), and a direct sum decomposition $\mathcal A_4\cong \mathbb{F}_4[Z]/\langle Z^4\rangle\oplus w\mathbb{F}_4[Z]/\langle Z^4\rangle$. The additive modular code $\mathcal D$ generated by the codeword $c=w+w^2 X_1+wX_1^2+w^2X_1^3$ can be seen as the submodule $\langle Z^2+wZ^3\rangle\le \mathcal A_4$. Magma computations \cite{sage} show that this code is self-orthogonal w.r.t. bilinear form considered in \cite[Section 4]{nuestro3}, and that the smallest weight of the codewords in $\mathcal D^\perp \setminus \mathcal D$ is 2. This provides the conditions to construct the $[[8,4,2]]$ quantum-error-correcting code mentioned above.}
\end{example}

{As an application of this theorem we will finally count the number of additive modular codes in PIPQRs.}

\begin{corollary}\label{conteocodigos}
Under the hypothesis of the previous theorem:
%Let $t_1(X)\in \mathbb{F}_2[X]$ and $t_2(Y)\in \mathbb{F}_2[Y]$ be monic polynomials such that $t_2$ is square-free. Then:
\begin{enumerate}
\item The total number of additive modular codes in $\mathcal A_4$ is
$$\prod_{C\in \mathcal C_2}\prod_{j=1}^s\prod_{m=1}^{s_{C,j}}\left(1+k_j+\left(\frac{2^{\delta^C_{j,m}}+1}{2^{\delta^C_{j,m}}-1}\right)\left(\frac{2^{k_j\delta^C_{j,m}}-1}{2^{\delta^C_{j,m}}-1}-k_j+2^{k_j\delta^C_{j,m}}-1\right)\right)$$
where $\delta^C_{j,m}=|C|\deg g^C_{j,m}(X_1)$.
\item Of these, only $$\prod_{C\in \mathcal C_2}\prod_{j=1}^s\prod_{m=1}^{s_{C,j}}\left(1+\left(\frac{2^{\delta^C_{j,m}}+1}{2^{\delta^C_{j,m}}-1}\right)(2^{k_j\delta^C_{j,m}}-1)\right)$$ codes can be generated by a single codeword.
\item Any of those codes can be generated at most by two codewords.
\end{enumerate}
\end{corollary}

\begin{proof}
Combine Theorem \ref{Aq} and \cite[Theorem 3]{nuestro5}.
\end{proof}

% \begin{acknowledgements}
% Edgar Mart\'inez-Moro was partially funded by Spanish MCINN under projects MTM2007-64704 and MTM2010-21580-C02-02. A. Pi\~nera-Nicol\'as and I. F. Rua were supported by MTM2010-18370-C04-01.
% \end{acknowledgements}
%\newpage

\bibliography{ring-semisimple}
\bibliographystyle{abbrv}

\end{document}